\documentclass[onecolumn, draftclsnofoot]{IEEEtran}
\usepackage{epsfig,graphics,epsf,color,amsmath,balance,cite,amsfonts,multirow,stfloats,algorithm,algorithmic,mathrsfs,amssymb,subfigure}
\usepackage[dvipdfm,unicode,colorlinks, linkcolor=black,anchorcolor=black,citecolor=black]{hyperref}
\newtheorem{theorem}{Theorem}

\newtheorem{proposition}{Proposition}

\newtheorem{proof}{Proof}

\begin{document}

\title{Energy Efficient Beamforming Optimization for Integrated Sensing and Communication}

\author{Zhenyao He,~Wei Xu,~\IEEEmembership{Senior Member,~IEEE,}~Hong Shen,~\IEEEmembership{Member,~IEEE,}\\~Yongming Huang,~\IEEEmembership{Member,~IEEE,} and Huahua Xiao

\thanks{
Z. He, W. Xu, and Y. Huang are with the National Mobile Communications Research Laboratory, Southeast University, Nanjing 210096, China, and also with the Purple Mountain Laboratories, Nanjing 211111, China (e-mail: \{hezhenyao, wxu, huangym\}@seu.edu.cn).

H. Shen is with the National Mobile Communications Research Laboratory, Southeast University, Nanjing 210096, China (e-mail: shhseu@seu.edu.cn).

H. Xiao is with the ZTE Corporation, and State Key Laboratory of Mobile Network and Mobile Multimedia Technology, Shenzhen 518055, China (e-mail: xiao.huahua@zte.com.cn).
}
}

\maketitle

\begin{abstract}
This paper investigates the optimization of beamforming design in a system with integrated sensing and communication (ISAC), where the base station (BS) sends signals for simultaneous multiuser communication and radar sensing. We aim at maximizing the energy efficiency (EE) of the multiuser communication while guaranteeing the sensing requirement in terms of individual radar beampattern gains.
The problem is a complicated nonconvex fractional program which is challenging to be solved.
By appropriately reformulating the problem and then applying the techniques of successive convex approximation (SCA) and semidefinite relaxation (SDR), we propose an iterative algorithm to address this problem. In theory, we prove that the introduced relaxation of the SDR is rigorously tight. Numerical results validate the effectiveness of the proposed algorithm.
\end{abstract}


\section{Introduction}
%Integrated sensing and communication (ISAC) has recently become a research hotspots, which integrates the radar sensing capabilities into wireless communication networks.
The requirements of high-precision sensing through wireless communication have recently attracted growing interests.
Integrated sensing and communication (ISAC), also known as joint radar-communication or dual-functional radar communications, is regarded as a promising solution to achieve this goal.
Compared to conventional systems of separate sensing and communication functions, ISAC brings tremendous advantages at cost reduction by spectrum resource sharing, hardware architecture reuse, and joint signal processing algorithms \cite{J.Zhang2021}.

For ISAC, it is of great importance to investigate the transmission design aiming at concurrent communication and radar sensing.
The authors of \cite{FLiu3} studied the problem of beamforming optimization in an ISAC system, where a waveform is reused for both radar sensing and data communication. This beamforming design minimized the beampattern matching error for sensing towards specific directions, subject to the constraints of minimum requirements of signal-to-interference-plus-noise ratio (SINR) at the users.
In \cite{M.Li2021}, the authors investigated a similar problem while utilizing the symbol-level precoding.
Alternatively, studies \cite{HHua} and \cite{XLiu} proposed to use dedicated signals that are irrelevant to the communication symbols for sensing. The proposed signals achieved enhanced performance because they allowed an additional degree of freedom (DoF) for sensing.
On the other hand, the problems of maximizing the transmitted beampattern gain towards sensing directions while guaranteeing the minimal SINR requirements of communication users were addressed in \cite{HHua} and \cite{full-duplex}.
In \cite{UAV}, the authors optimized the spectral efficiency in a UAV-assisted ISAC system, while ensuring the sensing requirements.
Moreover, a novel joint design of the transmitter and the radar receiver was addressed in \cite{C.Tsinos2021}, where the radar received SINR maximization is considered for the first time in ISAC.

With the rapidly growing demand for energy saving, energy efficiency (EE) has drawn increasing attentions as an important performance metric \cite{EE1,EE2,YeLiEE,A.Kaushik1,A.Kaushik2}.
For the purpose of reducing hardware complexity and power consumption, hybrid precoding was introduced in ISAC \cite{FLiu4} and the corresponding transmission optimization for EE maximization was studied in \cite{EEISAC}.
Specifically, by imposing the orthogonality assumption for the baseband precoder and using the Dinkelbach's method, the authors of \cite{EEISAC} first performed EE maximization with respect to the power allocation matrix. Then, the number of RF chains to be activated is determined based on the optimized matrix.
The authors in \cite{EEISAC2} proposed a successive convex approximation (SCA)-based algorithm to address the RF chain selection for EE maximization in ISAC.
However, these works do not directly address the EE maximized beamforming optimization problem.

This paper investigates an ISAC downlink system, where a base station (BS) simultaneously performs multiuser communication and radar target sensing by sending well-designed integrated signals.
Communication metric in terms of EE is the design goal for jointly optimizing the transmit beamforming, where the individual SINR requirements for user communications and the minimum beampattern gains towards sensing directions are treated as service constraints.
This problem is a complicated fractional program which is hard to be handled.
To solve this problem, we first derive an equivalent reformulation of the original problem and then adopt the techniques of SCA and semidefinite relaxation (SDR). Theoretically, we prove that the relaxation introduced by the SDR is tight.
The performance of the proposed algorithm is examined by the simulation results.

The rest of this paper is organized as follows. In Section \uppercase\expandafter{\romannumeral2}, the system model and the corresponding EE maximization problem are introduced. Section \uppercase\expandafter{\romannumeral3} presents the proposed algorithm. Simulation results are given in Section \uppercase\expandafter{\romannumeral4}. Finally, conclusions are provided in Section \uppercase\expandafter{\romannumeral5}.

\textit{Notations}: Vectors and matrices are denoted by boldface lower-case and boldface upper-case letters, respectively. $\mathbb C$ is the set of complex numbers. Superscripts $(\cdot)^T$ and $(\cdot)^H$ denote the matrix transpose and the conjugate transpose, respectively. $\text{Tr}(\cdot)$ denotes the trace of a matrix. The $\ell_2$ norm of a vector is denoted by $\|\cdot\|$. $\mathbb{E}\{\cdot\}$ returns the expectation and $\nabla$ is the gradient operator. $\text{rank}(\mathbf X)$ returns the rank of $\mathbf X$ and $\mathbf X \succeq \mathbf 0$ implies that $\mathbf X$ is positive semidefinite.

\section{System Model and Problem Formulation}
Consider an ISAC system consisting of a dual-functional BS and $K$ single-antenna users. The BS, equipped with an $N$-element uniform linear array (ULA),$\footnote{
The proposed algorithm is also applicable to the case where the BS is configured with uniform rectangular array or uniform circular array. }$
performs the downlink communication with $K$ users and conducts radar sensing towards $M$ potential target directions by simultaneously transmitting data signals $s_k,\ k=1,\cdots,K$ and a dedicated radar signal, $\mathbf s_0 \in \mathbb C^{N \times 1}$. These transmitted signals are written as
\begin{align}\label{def:x}
\mathbf x = \sum_{k=1}^K \mathbf v_k s_k + \mathbf s_0,
\end{align}
where $\mathbf v_k \in \mathbb C^{N \times 1}$ is the beamforming vector for user $k$ and $s_k$ satisfies a normalized power constraint, i.e., $\mathbb{E}\{ |s_k|^2 \} = 1$. The dual functional signal $\mathbf x$ is used for both communication and radar sensing.
Assume that the data signals intended to the users are mutually independent and are also uncorrelated with $\mathbf s_0$.
By defining $\mathbf V_0 \triangleq \mathbb{E}\{\mathbf s_0 \mathbf s_0^H\}$, the transmit power constraint at the BS is
$
\sum_{k=1}^K \| \mathbf v_k \|^2 + \text{Tr}(\mathbf V_0) \leq P_\text{max},
$
where $P_\text{max}$ denotes the budget of maximum transmit power.
%Consider the multi-beam transmission for the dedicated radar signal and thus $\mathbf V_0$ is with general rank.

Concerning the part of communication, the received radar signal is regarded as interference since it is unknown and meaningless to the users. Denote the channel from the BS to the $k$-th user by $\mathbf h_k \in \mathbb C^{N \times 1}$. Then, the SINR at user $k$ is expressed as
\begin{align}\label{def:gamma}
\gamma_k =  \frac{|\mathbf h_k^H \mathbf v_k |^2}{\sum_{i=1,i\neq k}^K |\mathbf h_k^H \mathbf v_i |^2 + \mathbf h_k^H \mathbf V_0 \mathbf h_k  + \sigma^2_k},
\end{align}
where $\sigma^2_k$ is the noise power at user $k$.

Concerning the radar sensing, the beamforming is expected to form a strong beampattern pointing to potential target directions of interest, thus yielding a larger radar receive signal-to-noise ratio (SNR) and better sensing performance \cite{MIMOradar2007,G.Cui2014}. Mathematically, the beampattern gain at a specific angular direction, $\theta$, is calculated as
\begin{align}
\!\! p(\theta) \!=\! \mathbb{E}\{ | \boldsymbol a^H(\theta) \mathbf x |^2 \}
\!=\! \boldsymbol a^H(\theta) \left(\sum_{k=1}^K  \mathbf v_k \mathbf v_k^H  +  \mathbf V_0 \right) \boldsymbol a(\theta),
\end{align}
where $\boldsymbol a(\theta) \triangleq  \frac{1}{N}[1,e^{j2\pi\frac{d}{\lambda}\sin(\theta)}, \cdots, e^{j2\pi\frac{d}{\lambda}(N-1)\sin(\theta)}]^T$
is the array steering vector of direction $\theta$, $d$ is the spacing between two adjacent antennas, and $\lambda$ represents the carrier wavelength.

Denote the system communication throughput by $R \triangleq \sum_{k=1}^K \log_2(1 + \gamma_k)$ and evaluate the total power consumption of the entire system as
\begin{align}\label{powermodel}
P_\text{total} = \frac{1}{\rho} \sum_{k=1}^K \| \mathbf v_k \|^2 + \frac{1}{\rho} \text{Tr}(\mathbf V_0) + P_{\rm c} + \xi R,
\end{align}
where $ \rho \in (0,1]$ stands for the amplifier efficiency of the BS, $P_{\rm c}$ is the summation of static circuit powers consumed by the BS and the users, and $\xi R$ represents the dynamic power with $\xi > 0$ denoting the dynamic power consumption per unit data rate\cite{YeLiEE}. Thus, the system EE is expressed as
\begin{align}
\eta_\text{EE} = \frac{ R }{\frac{1}{\rho} \sum_{k=1}^K \| \mathbf v_k \|^2 + \frac{1}{\rho} \text{Tr}(\mathbf V_0) + P_{\rm c} + \xi R}.
\end{align}

In this paper, we aim to jointly design the communication and radar beamforming, i.e., $\{\mathbf v_k\}_{k=1}^K$ and $\mathbf V_0$, at the BS to maximize $\eta_\text{EE}$. We consider that each user has an individual requirement of SINR, and the sensing performance is guaranteed by constraining minimal beampattern gains for $M$ directions.
By dividing the numerator and the denominator of $\eta_\text{EE}$ by $R$, it can be verified that the maximization of $\eta_\text{EE}$ is equivalent to the maximization of $\eta'_\text{EE} = \frac{R}{ \frac{1}{\rho} \sum_{k=1}^K \| \mathbf v_k \|^2 + \frac{1}{\rho} \text{Tr}(\mathbf V_0) + P_{\rm c}}$.
Accordingly, we formulate the EE maximization problem as
\begin{align}\label{Prob:ini}
\mathop \text{maximize} \limits_{\{\mathbf v_k\}_{k=1}^K, \mathbf V_0 \succeq \mathbf 0} \quad
&\eta'_\text{EE} = \frac{R}{ \frac{1}{\rho} \sum_{k=1}^K \| \mathbf v_k \|^2 + \frac{1}{\rho} \text{Tr}(\mathbf V_0) + P_{\rm c}}  \nonumber \\
\text{subject to}\quad
& \gamma_k \geq  \tau_k, \quad k = 1, \cdots, K \nonumber \\
& \sum_{k=1}^K \| \mathbf v_k \|^2 + \text{Tr}(\mathbf V_0) \leq P_\text{max}, \nonumber \\
& \boldsymbol a^H(\theta_m) \left(\sum_{k=1}^K  \mathbf v_k \mathbf v_k^H  +  \mathbf V_0 \right) \boldsymbol a(\theta_m) \geq \Gamma_m, \ m = 1, \cdots, M
\end{align}
where $\tau_k$ denotes the SINR requirement of user $k$, $\theta_m$ stands for the $m$-th target direction for radar sensing, and $\Gamma_m$ is the corresponding lowest beampattern gain for successful sensing. One can achieve the trade-off between the communication and the radar performances by adjusting the threshold $\Gamma_m$.

The above complicated EE maximization problem is a fractional program, which is challenging to be handled. In general, this kind of problem can be addressed based on the classical Dinkelbach's algorithm \cite{Dinkelbach} like in \cite{A.Kaushik2} and \cite{EEISAC}. As an alternative, in the following, we develop an SCA-based method to handle problem (\ref{Prob:ini}).

\section{Proposed Algorithm to Problem (\ref{Prob:ini})}
In this section, we first derive an equivalent reformulation to problem (\ref{Prob:ini}), and then apply the techniques of SCA and SDR to solve the reformulated problem. It is also proven that the relaxation of the SDR is rigorously tight.

\subsection{Problem Reformulation}
To reformulate (\ref{Prob:ini}) into a tractable form, we first introduce three auxiliary variables, i.e., $\mathbf V_k \triangleq \mathbf v_k \mathbf v_k^H$ for $k=1,\cdots,K$, $t$, and $u$.
The following proposition presents the equivalence between the reformulation and the original problem in (\ref{Prob:ini}).

\begin{proposition}\label{prop:tuW}
The problem in (\ref{Prob:ini}) is equivalently recast as
%\begin{subequations}
\begin{align}
\mathop \text{maximize}  \limits_{\{\mathbf V_k \succeq \mathbf 0\}_{k=0}^K,t,u} \quad
& t  \label{prob:refor} \\
\text{subject to}\quad
& u \geq \frac{1}{\rho} \sum_{k=1}^K \text{Tr}(\mathbf V_k) + \frac{1}{\rho} \text{Tr}(\mathbf V_0) + P_{\rm c}, \tag{\ref{prob:refor}{a}} \label{7b} \\
& t u \leq \sum_{k=1}^K \log_2(1+ \gamma_k') \tag{\ref{prob:refor}{b}}\label{C1},\\
& \mathbf h_k^H \mathbf V_k \mathbf h_k \geq  \tau_k \left(\sum_{i=0,i\neq k}^K \mathbf h_k^H  \mathbf V_i \mathbf h_k  + \sigma_k^2\right), \ k = 1, \cdots, K \tag{\ref{prob:refor}{c}}\label{7c}\\
& \sum_{k=1}^K \text{Tr}(\mathbf V_k) + \text{Tr}(\mathbf V_0) \leq P_\text{max}, \tag{\ref{prob:refor}{d}}\label{7d} \\
& \boldsymbol a^H(\theta_m) \left(\sum_{k=1}^K  \mathbf V_k +  \mathbf V_0 \right) \boldsymbol a(\theta_m) \geq \Gamma_m,  \ m = 1,\cdots, M \tag{\ref{prob:refor}{e}}\label{7e}\\
& t \geq 0, u \geq 0,  \tag{\ref{prob:refor}{f}} \label{7f}\\
& \text{rank}(\mathbf V_k) = 1, \quad k = 1, \cdots, K  \tag{\ref{prob:refor}{g}} \label{C2}
\end{align}
%\end{subequations}
where $\gamma_k'=\frac{\mathbf h_k^H \mathbf V_k \mathbf h_k}{\sum_{i=0,i\neq k}^K \mathbf h_k^H  \mathbf V_i \mathbf h_k  + \sigma_k^2}$ is an equivalent expression of $\gamma_k$ in (\ref{def:gamma}) by substituting the definition of $\mathbf V_k$.
\end{proposition}
\begin{proof}
See Appendix~\ref{proof:1}.
\end{proof}

The major advantage of reformulation (\ref{prob:refor}) is that the original fractional objective function in (\ref{Prob:ini}) is safely removed and most constraints in (\ref{prob:refor}) are of convex forms. In the following, we develop the efficient method to address problem (\ref{prob:refor}).

\subsection{Proposed Algorithm for Problem (\ref{prob:refor})}
The nonconvexity of problem (\ref{prob:refor}) lies in the constraint (\ref{C1}) and the rank-one constraints (\ref{C2}). To address these issues, we adopt the idea of SDR to relax (\ref{C2}) and apply the SCA method to deal with (\ref{C1}). We manipulate with the left-hand part and right-hand part of (\ref{C1}), separately.

Focusing on the left-hand part, we handle the nonconvex $f(t,u) \triangleq tu $ by invoking the SCA method of \cite{Glambda}, which allows us to replace $f(t,u)$ with its convex upper bound and then iteratively solve the sequence of resulting problems via judiciously updating the variables until
convergence.
%Focusing on the left-hand part, we replace the nonconvex $f(t,u) \triangleq tu $ by its convex upper bound based on the SCA method of \cite{Glambda}.
%By reasonably updating the variables until the convergence is reached, we can obtain a Karush-Kuhn-Tucker (KKT) point of (\ref{prob:refor}).
To do this, define the following function
\begin{align} \label{def:Glambda}
F (t,u,\lambda) \triangleq \frac{\lambda}{2} t^2 + \frac{1}{2 \lambda}  u^2,
\end{align}
which is convex to $(t,u)$ and overestimates $f(t,u)$ for every fixed $\lambda > 0$ since $F (t,u,\lambda) - f(t,u) = \frac{1}{2}( \sqrt{\lambda} t - \frac{1}{\sqrt{\lambda}}u)^2 \geq 0 $.
In particular, for $\lambda = \frac{u}{t}$, it holds that
\begin{align}\label{fF}
f(t,u) = F (t,u,\lambda), \
\nabla f(t,u)  = \nabla F (t,u,\lambda).
\end{align}
Thereby, we replace $f(t,u)$ by $F \left(t,u,\lambda^{(l-1)}\right)$ and update $\lambda^{(l)} = \frac{u^{(l)}} {t^{(l)}}$ in each iteration with $\left(u^{(l)},t^{(l)}\right)$ being the obtained solution to $(u,t)$ in the $l$-th iteration.

Concerning the right-hand part of (\ref{C1}), $\log_2(1+ \gamma_k')$ equals
\begin{align}
 \log_2 \left( \sum_{i=0}^K \mathbf h_k^H  \mathbf V_i \mathbf h_k  + \sigma_k^2 \right)
- \log_2 \left( \sum_{i=0,i\neq k}^K \mathbf h_k^H  \mathbf V_i \mathbf h_k  +  \sigma_k^2 \right),
\end{align}
where the function $\log(\cdot)$ is concave with respect to $ \{\mathbf V_k\}_{k=0}^K$ while the minus of the second term makes this constraint nonconvex.
We handle this nonconvexity using the SCA technique.
Specifically, by exploiting the concavity of $\log_2\left( \sum_{i=0,i\neq k}^K \mathbf h_k^H  \mathbf V_i \mathbf h_k  + \sigma_k^2 \right)$, we relax it by applying the first-order Taylor expansion. This gives a concave lower bound of $\log_2(1+ \gamma_k')$ in the $l$-th iteration as follows
\begin{align}\label{up:taylor}
 \log_2(1+ \gamma_k')
\geq  \log_2\left( \sum_{i=0}^K \mathbf h_k^H  \mathbf V_i \mathbf h_k  + \sigma_k^2 \right)
- \left( a_k^{(l-1)}  + \frac{\log_2e}{2^{a_k^{(l-1)}}} \sum_{i=0,i\neq k}^K \mathbf h_k^H  \left(\mathbf V_i - \mathbf V_i^{(l-1)} \right) \mathbf h_k \right)
\triangleq  \underline{R}_k,
\end{align}
where $ \{\mathbf V_i^{(l-1)}\}_{i=0}^K$ denotes the optimal solution obtained in the $(l-1)$-th iteration, and\\
$a_k^{(l-1)} = \log_2\left( \sum_{i=0,i\neq k}^K \mathbf h_k^H  \mathbf V_i^{(l-1)} \mathbf h_k + \sigma_k^2 \right)$.
%&b_k^{(l-1)} = \frac{\log_2e}{ \sum_{i=0,i\neq k}^K \mathbf h_k^H  \mathbf V_i^{(l-1)} \mathbf h_k + \sigma^2 }.

Now, applying the bounds in (\ref{def:Glambda}) and (\ref{up:taylor}), we acquire a sequence of surrogate problems to locally approximate problem (\ref{prob:refor}). In the $l$-th iteration, the surrogate convex optimization problem is formulated as
\begin{align}\label{prob:cvx}
\mathop \text{maximize}  \limits_{\{\mathbf V_k \succeq \mathbf 0\}_{k=0}^K,t,u} \quad
& t   \nonumber \\
\text{subject to}\quad
& F \left(t,u,\lambda^{(l-1)}\right) \leq \sum_{k=1}^K \underline{R}_k ,\nonumber \\
&  \rm (\ref{7b}), (\ref{7c}) \sim (\ref{7f}),
\end{align}
whose globally optimal solution can be found via, e.g., the interior point method \cite{cvx}.
The entire procedure of the developed SCA-based algorithm for solving (\ref{prob:refor}) is summarized in Algorithm \ref{alg:SCA}.
As shown in (\ref{fF}), the nonconvex original function, $f(t,u)$, and its convex approximation, $F(t,u,\lambda)$, have the identical value and gradient when $\lambda = \frac{u}{t}$. So do $\log_2(1+\gamma_k')$ and its approximation, $\underline{R}_k$, in (\ref{up:taylor}) when $\{\mathbf V_i^{(l-1)}\}_{i=0}^k = \{\mathbf V_i\}_{i=0}^k$. Thus,
according to \cite{Glambda}, the iterative procedure can converge to a Karush-Kuhn-Tucker (KKT) point of problem (\ref{prob:refor}) (omitting the rank-one constraints (\ref{C2})). Moreover, the main computational burden in Algorithm 1 stems from solving problem ({\ref{prob:cvx}}), whose complexity is $\mathcal O (N^{6.5} K^{3.5} + N^{4} K^{2} M^{1.5} )$ \cite{complexity}.

\begin{algorithm}[t]
\caption{SCA-Based Algorithm for Problem (\ref{prob:refor})}
\label{alg:SCA}
\begin{algorithmic}[1]
\STATE \textit{Initialization:} Set initial point $\{\mathbf{V}_k^{(0)}\}_{k=0}^K$, $\lambda^{(0)}$, iteration index $l = 0$, and convergence accuracy $\epsilon$.
\REPEAT
    \STATE Set $l = l + 1$.\\
    \STATE
    Solve problem (\ref{prob:cvx}) with $ \{\mathbf V_k^{(l-1)}\}_{k=0}^K$ and $\lambda^{(l-1)}$.\\

    \STATE Update $ \{\mathbf V_k^{(l)}\}_{k=0}^K$ and $\lambda^{(l)} = \frac{u^{(l)}}{t^{(l)}}$.

\UNTIL convergence.
%\STATE Compute $\{ \mathbf {\widehat V}_k\}_{k=0}^K$ according to (\ref{def:hat}) by using $\mathbf {\widetilde V}_k = \mathbf V_k^{(l)}$ for $k = 0,\cdots,K$.
%\STATE Perform Cholesky decomposition to $\{ \mathbf {\widehat V}_k\}_{k=1}^K$ to get the optimal beamforming vector $\{ \mathbf v_k^*\}_{k=1}^K$.
%\STATE Obtain the optimal covariance matrix for sensing signal, $\mathbf V^*_0 = \mathbf {\widehat V}_0$.
\end{algorithmic}
\end{algorithm}

Note that the remaining issue is that the rank-one constraints in problem (\ref{prob:refor}) are relaxed via the SDR and the resulting problem in (\ref{prob:cvx}) may not return a rank-one solution. In general, an additional procedure like the Gaussian randomization \cite{GaussianR} is used to find a rank-one solution, which usually leads to high computational complexity and the optimality is not guaranteed. Fortunately, we show in the following theorem that there always exists a rank-one optimum for problem (\ref{prob:cvx}).

\begin{theorem}\label{prop:rank1}
There always exists a global optimum to (\ref{prob:cvx}), denoted by $\{ \mathbf {\widehat V}_k\}_{k=0}^K$, satisfying the rank-one constraints:
\begin{align}
\text{rank}(\mathbf {\widehat V}_k) = 1, \ k=1,\cdots,K.
\end{align}
\end{theorem}
\begin{proof}
See Appendix~\ref{proof:2}.
\end{proof}
%With the preparations above, the entire procedure of the developed algorithm for solving (\ref{Prob:ini}) is readily summarized in Algorithm \ref{alg:SCA}.
Based on the above theorem, we could construct $\{ \mathbf {\widehat V}_k\}_{k=0}^K$ according to (\ref{def:hat}) when Algorithm \ref{alg:SCA} converges, and then calculate the optimal beamforming vectors $\{ \mathbf {v}^*_k\}_{k=1}^K$ by performing Cholesky decomposition to the rank-one $\{ \mathbf {\widehat V}_k\}_{k=1}^K$.

\begin{table}[t]
\caption{Simulation Parameters}
\label{table1}
    \centering
    \setlength\tabcolsep{1.5mm}{
	\begin{tabular}{|c|c|c|}
		\hline \textbf{Notation} &\textbf{Parameter} & \textbf{Value} \\
        \hline $N$ & Number of antennas at the BS  & 16\\
        \hline $K$ & Number of users & 2 \\
        \hline $M$ & Number of sensing directions & 4 \\
        \hline $P_\text{max}$ & Maximum transmit power at the BS  & 30 dBm\\
        \hline $P_{\rm c}$ & Circuit power of the system & 25 dBm\\
        \hline $\sigma^2$ & Noise power at the users  & -80 dBm \\
        \hline - & Path loss from the BS to each user & -99 dB \\
        \hline $\tau$ & Required minimum SINR at the users  &  5 dB \\
        \hline $\Gamma$ & Required radar beampattern gain &  20 dBm \\
        \hline $\rho$ & Amplifier efficiency at the BS  & 0.35 \\
        \hline $\xi$ & Dynamic power consumption coefficient &  -26 dBm/bps \cite{YeLiEE} \\
        \hline $\epsilon$ & Algorithm convergence accuracy &  0.001 \\
        \hline
	\end{tabular}%
}
\end{table}

\section{Simulation Results}
The performance of the proposed algorithm is verified via numerical simulations. Assume that the ULA at the BS owns half-wavelength spacing between adjacent antennas, i.e., $d = \lambda/2$, and $M=4$ sensing directions are $-54^{\circ}$, $-18^{\circ}$, $18^{\circ}$, and $54^{\circ}$. The channels from the BS to the users follow the line-of-sight (LoS) model, where the angle of departure (AoD) from the BS to user $k$ is denoted by $\phi_k$. We set $K=2$, $\phi_1 = -30^{\circ}$, and $\phi_2= 30^{\circ}$.
The noise power and the SINR thresholds for different users are set to $\sigma^2 = \sigma^2_k$ and $\tau = \tau_k,\ k = 1,\cdots,K$, and the beampattern gains at all target directions are set to $\Gamma = \Gamma_m,\ m = 1,\cdots,M$. Specific values of the simulation parameters are listed in Table \ref{table1}.

\begin{figure}[t]
\begin{center}
      \epsfxsize=7.0in\includegraphics[scale=0.6]{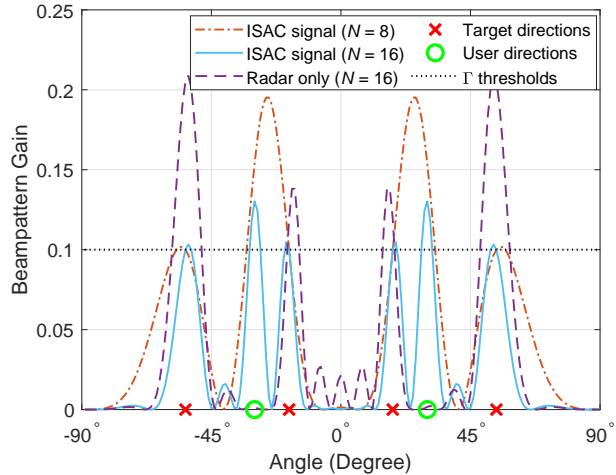}
      \caption{Beampattern gain achieved by the proposed algorithm.}\label{fig:angledegree}
    \end{center}
\end{figure}

\begin{figure}[t]
\centering
\subfigure[Detection probability]{
\begin{minipage}[t]{0.5\linewidth}
\centering
\includegraphics[width=3in]{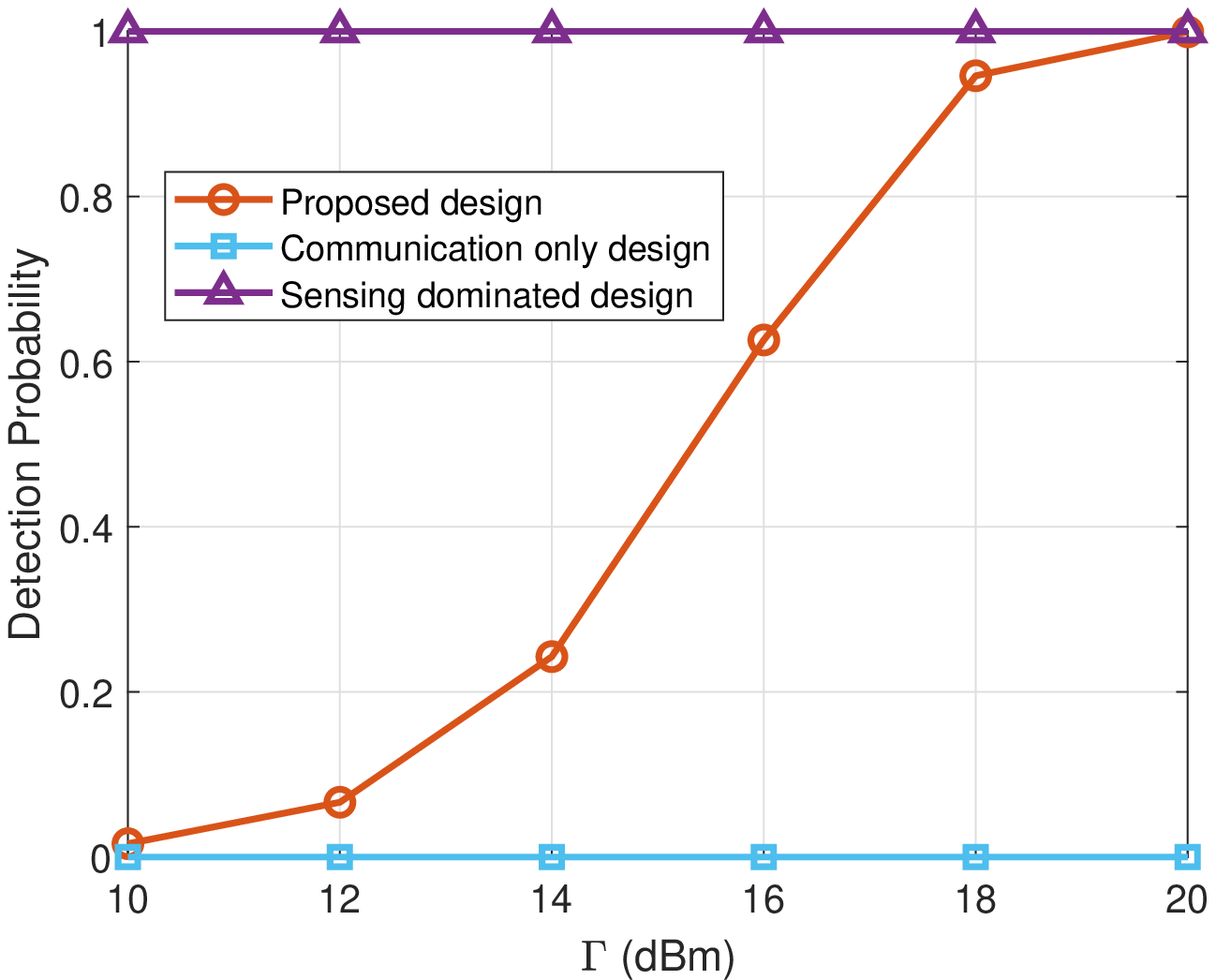}
%\caption{fig1}
\end{minipage}%
}%
\subfigure[EE performance]{
\begin{minipage}[t]{0.5\linewidth}
\centering
\includegraphics[width=3in]{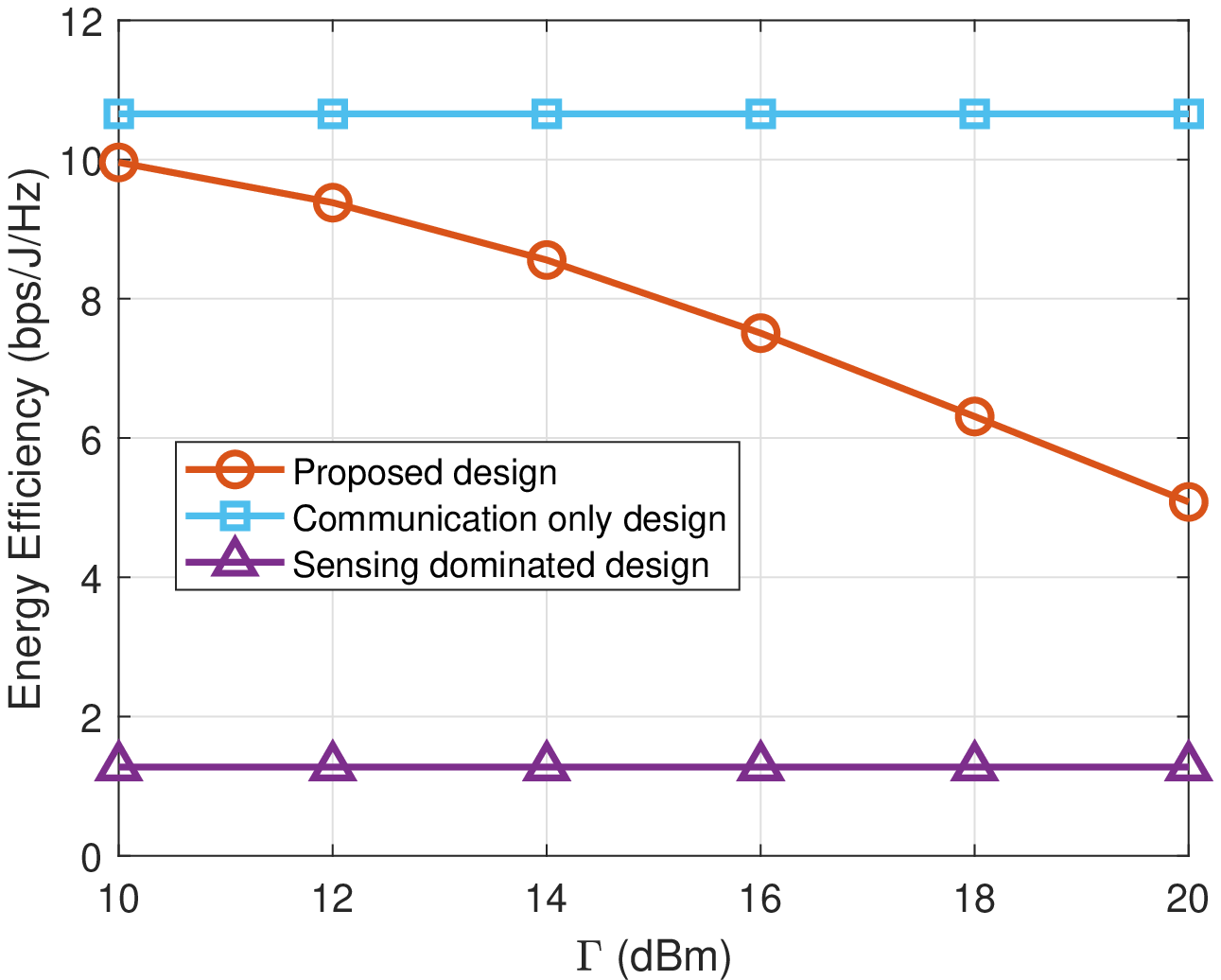}
%\caption{fig2}
\end{minipage}%
}%
\centering
\caption{Detection probability and EE performance versus required radar beampattern gain.}\label{fig:gamma}
\end{figure}

%\begin{figure}[t]
%\subfloat[Detection probability]{\includegraphics[width=0.24\textwidth]{detectionprob.eps}}
%\hfill
%\subfloat[EE performance]{\includegraphics[width=0.24\textwidth]{fig2.eps}}
%      %\caption{}\label{fig:angledegree}
%\end{figure}

%\begin{figure}[t]
%\begin{center}
%      \epsfxsize=7.0in\includegraphics[scale=0.4]{fig2.eps}
%      \caption{EE performance versus required radar beampattern gain.}\label{fig:sumrate}
%    \end{center}
%\end{figure}

Fig. \ref{fig:angledegree} shows the beampattern gains achieved by the proposed design. It is clearly seen that the optimized ISAC design allocates the beampattern towards desired directions of the communication users and the sensing targets and the gains towards target directions are all larger than the threshold $\Gamma$, which indicates the effectiveness of our proposed algorithm.
%Moreover, by comparing the two curves with $N=8$ and $N=16$, it is observed that a larger ULA is capable of achieving narrower beams.

Fig. \ref{fig:gamma} depicts the detection probability and EE performance of the proposed design. We consider two benchmark schemes for comparison.
The first one is a communication only scheme, which is obtained by solving problem (\ref{prob:refor}) without the sensing constraints in (\ref{7e}). The second one is a sensing dominated design, which is obtained by maximizing the transmit beampattern gain towards sensing directions \cite{MIMOradar2007}, while guaranteeing the minimal SINR requirements of communication users. We calculate the detection probability for point-like target detection according to \cite{C.Tsinos2021} and \cite{detection}, where the power of the target and the noise variance are set to 25 dBW and 0 dBW, respectively, and the false alarm probability is set to $10^{-5}$. From Fig. \ref{fig:gamma} (a), it can be found that the sensing performance in terms of detection probability can be guaranteed by setting a relatively large $\Gamma$. From Fig. \ref{fig:gamma} (b), it can be found that, as $\Gamma$ increases, the achieved EE of the proposed design decreases. This is because, with the growth of $\Gamma$, more transmit power should be used for guaranteeing the sensing performance and meanwhile the communication rate is compromised.

\section{Conclusion }
We investigated the EE maximization problem of an ISAC system, where the BS takes radar sensing into account while keeping communication with users. The complicated fractional program was transformed into a tractable form and further iteratively addressed by adopting the techniques of SCA and SDR. We also proved that the introduced relaxation is tight. The effectiveness of our proposed algorithm was verified through the simulation results.

\begin{appendices}
\section{Proof of Proposition ~\ref{prop:tuW}}\label{proof:1}
In order to guarantee the equivalence between problem (\ref{prob:refor}) and problem (\ref{Prob:ini}), the key is to verify that the constraints (\ref{7b}) and (\ref{C1}) in problem (\ref{prob:refor}) must keep active at the optimality. We complete the proof by contradiction.

Specifically, assume that an optimal solution to problem (\ref{prob:refor}) is obtained as $(\{\mathbf V'_k\}_{k=0}^K,t',u')$, while at least one of (\ref{7b}) and (\ref{C1}) are inactive at this optimality. In such cases, we can always decrease the value of $u'$ and subsequently increase $t'$ or directly increase $t'$, thus yielding a larger objective.
%This includes the following three cases:
%\begin{itemize}
%  \item Case 1: $u' > P'_\text{total} $ and $ t' u' < \sum_{k=1}^K \log_2(1+ \gamma_k')$,
%  \item Case 2: $u' > P'_\text{total} $ and $ t' u' = \sum_{k=1}^K \log_2(1+ \gamma_k')$,
%  \item Case 3: $u' = P'_\text{total} $ and $ t' u' < \sum_{k=1}^K \log_2(1+ \gamma_k')$,
%\end{itemize}
%where $P'_\text{total} = \frac{1}{\rho} \sum_{k=1}^K \text{Tr}(\mathbf V'_k) + \frac{1}{\rho} \text{Tr}(\mathbf V'_0) + P_{\rm c}$.
%For Cases 1 and 2, we can always decrease the value of $u'$ and subsequently increase $t'$ until equalities are established for these inactive constraints, which definitely achieves a larger objective function. For Case 3, we can directly increase $t'$ such that $ t' u' = \sum_{k=1}^K \log_2(1+ \gamma_k')$.
These contradict with the fact that $t'$ is maximized at this solution.
Therefore, we conclude that (\ref{Prob:ini}) and (\ref{prob:refor}) are equivalent.

\section{Proof of Theorem ~\ref{prop:rank1}}\label{proof:2}
Let $\{\mathbf {\widetilde V}_k\}_{k=0}^K$ be an arbitrary optimal solution to problem (\ref{prob:cvx}). Then, we can construct a new solution, $\{\mathbf {\widehat V}_k\}_{k=0}^K$, based on $\{\mathbf {\widetilde V}_k\}_{k=0}^K$ as
\begin{align}\label{def:hat}
\mathbf {\widehat V}_k =& \mathbf {\widehat v}_k \mathbf {\widehat v}_k^H, \ k=1,\cdots,K \nonumber  \\
\mathbf {\widehat V}_0 =& \sum_{k=1}^K \mathbf {\widetilde V}_k + \mathbf {\widetilde V}_0 - \sum_{k=1}^K \mathbf {\widehat V}_k,
\end{align}
where $\mathbf {\widehat v}_k = (\mathbf h_k^H \mathbf {\widetilde V}_k \mathbf h_k )^{-1/2} \mathbf {\widetilde V}_k \mathbf h_k, \ k=1,\cdots,K$.

It is straightforwardly seen that $\mathbf {\widehat V}_k,\ k=1,\cdots,K$ is rank-one and positive semidefinite. In the following, we prove that $\{\mathbf {\widehat V}_k\}_{k=0}^K$ is also feasible and optimal to problem (\ref{prob:cvx}).

First, we verify that $\mathbf {\widehat V}_0 \succeq \mathbf 0 $. Specifically, for an arbitrary $\mathbf f \in \mathbb C^{N \times 1}$, we have
\begin{align} \label{positive}
\mathbf f^H (\mathbf {\widetilde V}_k - \mathbf {\widehat V}_k) \mathbf f = \mathbf f^H \mathbf {\widetilde V}_k \mathbf f
- (\mathbf h_k^H \mathbf {\widetilde V}_k \mathbf h_k )^{-1} | \mathbf f^H \mathbf {\widetilde V}_k \mathbf h_k |^2
\overset{(\rm a)}{ \geq } \mathbf f^H \mathbf {\widetilde V}_k \mathbf f  - (\mathbf h_k^H \mathbf {\widetilde V}_k \mathbf h_k )^{-1} | \mathbf f^H \mathbf {\widetilde v}_k |^2 |\mathbf {\widetilde v}_k \mathbf h_k |^2
=  0,
\end{align}
where step (a) follows from the Cauchy-Schwarz inequality. Thus, it holds that $\mathbf {\widetilde V}_k - \mathbf {\widehat V}_k \succeq \mathbf 0$ and $\sum_{k=1}^K \mathbf {\widetilde V}_k - \sum_{k=1}^K \mathbf {\widehat V}_k \succeq \mathbf 0$. Since $\mathbf {\widetilde V}_0 \succeq \mathbf 0$, it then follows from (\ref{def:hat}) that $\mathbf {\widehat V}_0 \succeq \mathbf 0$.

Next, based on (\ref{def:hat}), we have the following equality
\begin{align}\label{sum1}
\sum_{k=1}^K \mathbf {\widetilde V}_k + \mathbf {\widetilde V}_0 = \sum_{k=1}^K \mathbf {\widehat V}_k + \mathbf {\widehat V}_0,
\end{align}
which implies that the constraints (\ref{7b}), (\ref{7d}), and (\ref{7e}) hold. Moreover, it can be easily shown that
\begin{align}\label{sum2}
\mathbf h_k^H \mathbf {\widehat V}_k \mathbf h_k = \mathbf h_k^H \mathbf {\widehat v}_k \mathbf {\widehat v}_k^H \mathbf h_k = \mathbf h_k^H \mathbf {\widetilde V}_k \mathbf h_k.
\end{align}
Then, based on (\ref{sum1}) and (\ref{sum2}), the SINR constraint in (\ref{7c}) also holds for $\{\mathbf {\widehat V}_k\}_{k=0}^K$ since it can be rewritten as
$
\frac{1+\tau_k}{\tau_k} \mathbf h_k^H \mathbf {\widetilde V}_k \mathbf h_k \geq \sum_{i=0}^K\mathbf h_k^H  \mathbf {\widetilde V}_i \mathbf h_k + \sigma^2_k.$

The last step is to verify that the objective value achieved by $\{\mathbf {\widehat V}_k\}_{k=0}^K$ is identical to that by $\{\mathbf {\widetilde V}_k\}_{k=0}^K$. Equivalently, this is to prove that the value of $\underline{R}_k$ in the first constraint in (\ref{prob:cvx}) remains unchanged. To this end, we reexpress the terms related to $\{\mathbf {\widetilde V}_k\}_{k=0}^K$ in $\underline{R}_k$ from (\ref{up:taylor}) as
\begin{align}\label{sum3}
\log_2\left( \sum_{i=0}^K \mathbf h_k^H  \mathbf {\widetilde V}_i \mathbf h_k  + \sigma^2_k \right)
 -  \frac{\log_2e}{2^{a_k^{(l-1)}}}  \left( \sum_{i=0}^K\mathbf h_k^H  \mathbf {\widetilde V}_i \mathbf h_k  -  \mathbf h_k^H \mathbf {\widetilde V}_k \mathbf h_k \right).
\end{align}
Due to the equalities in (\ref{sum1}) and (\ref{sum2}), the expression in (\ref{sum3}) remains the same when substituting $\{\mathbf {\widehat V}_k\}_{k=0}^K$ to it.

Thus far, it is verified that the constructed $\{\mathbf {\widehat V}_k\}_{k=0}^K$ in (\ref{def:hat}) is a feasible solution and also a global optimum to problem (\ref{prob:cvx}), which completes the proof.
\end{appendices}

\end{document}